\documentclass[11pt]{article}
\usepackage{fullpage}
\usepackage{comment}
\usepackage{algorithm,algorithmic}

\usepackage{graphicx}

\usepackage[english]{babel}
\usepackage[utf8]{inputenc}

\usepackage{amsmath}
\usepackage{amsthm}
\usepackage{amssymb}
\usepackage{graphicx}
\usepackage{xcolor}
\usepackage{enumerate}
\usepackage[inline]{enumitem}
\SetLabelAlign{Center}{\hfil#1\hfil}
\SetLabelAlign{CenterWithParen}{(\makebox[1.0em]{#1})}
\usepackage{ifpdf}
\ifpdf    
\usepackage{hyperref}
\else    
\usepackage[hypertex]{hyperref}
\fi

\newtheorem{lemma}{Lemma}[section]
\newtheorem{definition}[lemma]{Definition}
\newtheorem{theorem}[lemma]{Theorem}

\newtheorem{observation}[lemma]{Observation}
\newtheorem{corollary}[lemma]{Corollary}

\makeatletter
\newcommand{\ovee}{\mathbin{\mathpalette\make@circled\vee}}
\newcommand{\make@circled}[2]{%
  \ooalign{$\m@th#1\smallbigcirc{#1}$\cr\hidewidth$\m@th#1#2$\hidewidth\cr}%
}
\newcommand{\smallbigcirc}[1]{%
  \vcenter{\hbox{\scalebox{0.82}{$\m@th#1\bigcirc$}}}%
}

\include{pythonlisting}
\newcommand{\canonical}{{canonical}}
\newcommand{\weightdirect}{{weighted directed}}
\newcommand{\minmax}{{(min,max)\text{-}product}}
\newcommand{\tminmax}{{target\text{-}(min,max)\text{-}product}}
\newcommand{\Tminmax}{{Target\text{-}(min,max)\text{-}product}}
\newcommand{\tminmaxShortcut}{{T\text{-}(min,max)\text{-}product}}

\newcommand{\AlgoAPSPtoMINMAX}[2]{\texttt{APSP-To-TMinMax}$(#1,#2)$}
\newcommand{\name}{{restricted T\text{-}(min,max)\text{-}product}}
\newcommand{\Name}{{Restrictred T\text{-}(min,max)\text{-}product}}
\usepackage{authblk}

\title{$\{-1,0,1\}$-APSP and (min,max)-Product Problems\footnote{This work is partially supported by ISF grants (no. 1278/16 and 1926/19),  by a grant from the United States - Israel Binational Science Foundation (BSF) (no. 2018364), and by an ERC grant MPM under the European Union's Horizon 2020 research and innovation programme (no. 683064).}}
\author[ ]{Hodaya Barr}
\author[ ]{Tsvi Kopelowitz}
\author[ ]{Ely Porat}
\author[ ]{Liam Roditty}
\affil[ ]{Bar-Ilan University, Ramat Gan, Israel}
\affil[ ]{\texttt{odayaben@gmail.com, kopelot@gmail.com}}
 \affil[ ]{ \texttt{porately@cs.biu.ac.il, liamr@macs.biu.ac.il}}

\date{}

\begin{document}
\maketitle
  \algsetup{
    linenodelimiter={.}
  }

\begin{abstract}
In the $\{-1,0,1\}$-APSP problem the goal is to compute  all-pairs shortest paths (APSP) on a directed graph whose edge weights are all from $\{-1,0,1\}$.
In the (min,max)-product problem the input is two $n\times n$ matrices $A$ and $B$, and the goal is to output the (min,max)-product of $A$ and $B$.

This paper provides a new algorithm for the $\{-1,0,1\}$-APSP problem via a simple reduction to the target-(min,max)-product problem where the input is three $n\times n$ matrices  $A,B$, and $T$, and the goal is to output a Boolean $n\times n$ matrix $C$ such that the $(i,j)$ entry of $C$ is 1 if and only if the $(i,j)$ entry of the (min,max)-product of $A$ and $B$ is exactly the $(i,j)$ entry of the target matrix $T$.
If (min,max)-product can be solved in $T_{MM}(n) = \Omega(n^2)$ time then it is straightforward to solve target-(min,max)-product in $O(T_{MM}(n))$ time.
Thus, given the recent result of Bringmann, Künnemann, and Wegrzycki [STOC 2019], the $\{-1,0,1\}$-APSP problem can be solved in the same time needed for solving approximate APSP on graphs with positive weights.

Moreover, we design a simple algorithm for target-(min,max)-product 
when the inputs are restricted to the family of inputs generated by our reduction. Using fast rectangular matrix multiplication, the new algorithm is faster than the current best known algorithm for (min,max)-product.

\end{abstract}

\newpage
\section{Introduction}
The \emph{all-pairs shortest paths} (APSP) problem is one of the most fundamental algorithmic problems in computer science, and is defined as follows.
Let $G=(V,E)$  be a weighted (directed) graph  with weight function $w: E\rightarrow\mathbb{R}$, where $V=\{v_1,v_2,\ldots, v_n\}$.
Let $A$ be the weighted adjacency matrix of $G$.
For $v_i,v_j\in V$, the weight of the shortest path between $v_i$ and $v_j$ is denoted by $a^*_{ij}$, and the matrix containing the weights of the shortest paths between all pairs of vertices is denoted by $A^*$.
In the APSP problem the goal is to compute $A^*$.

The Floyd-Warshall APSP algorithm~\cite{floyd62,roy1959transitivite,warshall1962theorem}, is a dynamic programming algorithm, whose runtime is $O(n^3)$.
After several improvements of poly-logarithmic factors in the time cost~\cite{fredman76,takaoka91, Han04, Takaoka04, Takaoka05, Zwick04, Chan05,Han06, chan2006all,Chan10,HanT12}, in a recent breakthrough, Williams~\cite{Williams18} and Chan and Williams~\cite{CW16} designed the current best algorithm whose time cost is $O(\frac{n^3}{2^{\Omega(\sqrt{\log n})}})$.
The lack of success in designing a truly sub-cubic algorithm for APSP has led to a popular conjecture that any algorithm for APSP requires $\Omega(n^{3-o(1)})$ time~\cite{RZ11,WW10}.
However, in some special cases, faster algorithms are known.
We describe two examples which are strongly related to the results in this paper.
Throughout the paper, let $\omega$ denote the exponent in the fastest fast matrix multiplication (FMM) algorithm; the current best upper bound on $\omega$ is roughly  $2.3728639$~\cite{Gall14a}.

\paragraph{Approximate APSP with positive real weights and \minmax{}.}
The first example is the \emph{approximate positive APSP} problem where the goal is to compute a $1\pm \epsilon$  approximation of the distances for a graph with positive weights.
A recent result by Bringmann, Künnemann, and Wegrzycki~\cite{BKW19} shows that approximate positive APSP can be solved in $\tilde{O}\left(\frac{n^{\frac{3+\omega}{2}}}{poly(\epsilon)}\right)=\tilde{O}\left(\frac{n^{2.686}}{poly(\epsilon)}\right) $ time\footnote{We use the standard $\tilde O$ notation to suppress poly-logarithmic factors.}.
The algorithm in~\cite{BKW19} is obtained by showing an equivalance between approximate positive APSP and  the \minmax{} problem which is defined as follows.
Throughout the paper we follow the notion that matrices are denoted with capital letters, while the entries of matrices are denoted using the same letter, just lowercase, with the appropriate indices indicated as subscripts.
So, for example the $(i,j)$ entry of matrix $A$ is denoted by $a_{ij}$.

\begin{definition}[\minmax{}]
In the \emph{\minmax{}} problem the input is two matrices $A\in \mathbb{R}^{n\times n}$ and $B \in \mathbb{R}^{n\times n}$, and the output is a matrix $C\in\mathbb{R}^{n\times n}$ such that the $(i,j)$ entry of $C$ is
$$c_{ij} = \min_{k=1}^{n} \{\max(a_{ik}, b_{kj})\}.$$
We denote the \minmax{} between $A$ and $B$ with the
$\ovee$ operator; that is, $C = A\ovee B$.
\end{definition}
Duan and Pettie~\cite{DP09} showed that the \minmax{} problem can be solved in $O(n^{(3+\omega)/2})$ time.

\paragraph{$\{-1,0,1\}$-APSP.}
The second example is the $\{-1,0,1\}$\emph{-APSP} problem where the weights are from the set $\{-1,0,1\}$.
Negative edge weights introduce a new depth to the challenges involved in solving APSP, and
the $\{-1,0,1\}$-APSP problem is perhaps the purest version of APSP that allows for negative weights.

Alon, Galil and Marglit~\cite{AGM97} designed an algorithm for the $\{-1,0,1\}$\emph{-APSP} problem  whose runtime is $O(n^{\frac{3+\omega}{2}})=O(n^{2.686})$.
Zwick~\cite{zwick02} improved the runtime to $O(n^{2+\mu}) = O(n^{2.5286})$, where $\mu$ satisfies $\omega(1,\mu,1) = 1+2\mu$, and $\omega(1,\mu,1)$ is the exponent of $n$ in the fastest algorithm for multiplying two rectangular matrices of sizes $n\times n^{\mu}$ and $n^\mu \times n$.
The current best upper bound on $\mu$ is roughly $0.5286$~\cite{Gall14a}.

\subsection{Our Results and Algorithmic Overview}
In this paper we introduce a reduction from $\{-1,0,1\}$-APSP to \minmax{}, which combined with~\cite{BKW19} implies that $\{-1,0,1\}$-APSP is reducible to approximate positive APSP.
Specifically, we use the framework of Alon et al.~\cite{AGM97} who follow the paradigm of Seidel's APSP algorithm for unweighted undirected graphs~\cite{Seidel95}.
Seidel's algorithm 
first recursively solves the undirected unweighted APSP problem on a specially constructed graph $G'$ so that if the distance between vertices $v_i$ and $v_j$ in the original graph is $a^*_{ij}$, then the distance between $v_i$ and $v_j$ in $G'$ is $t^*_{ij}$ where $t^*_{ij} = \Big\lceil \frac{a^*_{ij}}{2}\Big\rceil$.
Then, the problem of computing $a^*_{ij}$ reduces to the problem of establishing whether $a^*_{ij}$ is odd or even.
The algorithm of~\cite{AGM97} follows the same structure as Seidel's algorithm, but instead of directly computing the parity of $a^*_{ij}$, the algorithm uses a brute-force like method.

Instead, we use a more direct approach for computing the parity of $a^*_{ij}$ by solving two instances of the \emph{\tminmax{} problem} where the input is the same as the input for the \minmax{}  problem together with a third \emph{target matrix} $T\in \mathbb{R}^{n\times n}$.
The output is a Boolean matrix whose $(i,j)$ entry is  an indicator of whether the $(i,j)$ entry of the \minmax{} is equal to the $(i,j)$ entry in the target matrix. Formally:

\begin{definition}[\Tminmax{}]\label{def:t-minmax}
 In the \emph{\tminmax{}} (\tminmaxShortcut{}) problem the input is three matrices $A\in \mathbb{R}^{n\times n}$, $B \in \mathbb{R}^{n\times n}$ and $T\in\mathbb{ R}^{n\times n}$, and the output is a Boolean matrix $C \in \{0,1\}^{n\times n}$ such that $c_{ij} = 1$ if and only if $t_{ij} = \min_{k=1}^{n}\{\max(a_{ik}, b_{kj})\}$. The \tminmaxShortcut{} operation is denoted by $C = A \ovee_T B$.
\end{definition}
Given three matrices $A,B$ and $T$ it is straightforward to compute the \tminmaxShortcut{} of $A,B$ and $T$ by first computing $A\ovee B$ and then spending another $O(n^2)$ time to compare each entry $A\ovee B$ with the corresponding entry in $T$.

Our main result is summarized by the following theorem (the proof is given in Section~\ref{sec:reduction}).

\begin{theorem}\label{thm:main}
Suppose that there exists an algorithm for solving the \tminmaxShortcut{} problem in $T_{TMM}(n)\ge n^2$ time.
Then there exists an algorithm for solving $\{-1,0,1\}$-APSP in $\tilde O(T_{TMM}(n)+n^\omega)$ time.
\end{theorem}

\paragraph{A simple algorithm for \name{}.}
In addition to showing that $\{-1,0,1\}$-APSP is reducible to $\tminmaxShortcut{}$ (and therefore also to \minmax{}), we also design a simple algorithm for solving \tminmaxShortcut{} for the restricted family of inputs generated by our reduction.
Specifically, the
 entries in the second matrix $B$ are $\pm \infty$, and the target matrix $T$ has the property that for any\footnote{We use the standard notation that $[n] = \{1,2,3,\ldots,n\}$.} $i,j\in [n]$, $t_{ij}\le \min_{k=1}^{n} \{\max(a_{ik}, b_{kj})\}$.
Formally, the \tminmaxShortcut{} problem on this family of inputs is defined as follows.
\begin{definition}[\name{}] \label{def:restricted-T-minmax}
In the \emph{\name{} problem} the input is composed of a matrix $A \in \mathbb{R}^{n\times n}$, a matrix $B \in \{-\infty,\infty\}^{n\times n}$,
and a target matrix $T \in \mathbb{R}^{n\times n}$, where for every $i,j\in [n]$, $t_{ij} \leq  \min_{k=1}^{n} \{\max(a_{ik}, b_{kj})\}$.
The output is a matrix $C \in \{0,1\}^{n\times n}$, such that for every $i,j\in [n]$, $c_{ij} = 1$ if and only if $t_{ij} = \min_{k=1}^{n} \{\max(a_{ik}, b_{kj})\}$.
\end{definition}

In Section~\ref{sec:restricted-algo} we prove the following theorem.

\begin{theorem}\label{thm:restricted}
For any $0\le t\le 1$, there exists an algorithm for the \name{} problem whose time cost is $\Tilde{O}(n^{2+t} +n^{\omega(2-t,1,1)})$ time.
\end{theorem}

Using the algorithm of \cite{GU18} for fast rectangular matrix multiplication,  we are able to upper bound $\Tilde{O}(n^{2+t} +n^{\omega(2-t,1,1)})$ by $\tilde O(n^{2.658044})$; see Appendix~\ref{app:run-time} for a detailed explanation.
Thus, using our reduction from Theorem~\ref{thm:main} together with the algorithm of Theorem~\ref{thm:restricted}, we obtain a new algorithm for $\{-1,0,1\}$-APSP whose cost is $O(n^{2.658044})$ time. 
Notice that using fast squared matrix multiplication instead of fast rectangular matrix multiplication, the time bound becomes $\tilde O(n^{\frac{3+\omega}{2}})=\tilde O(n^{2.686})$, which matches the runtime  of~\cite{DP09} for solving \minmax{}.

We remark that Zwick's algorithm~\cite{zwick02} is faster and has a time cost of $\tilde O(n^{2.5286})$.
Nevertheless, if $\omega=2$ then the time costs of our algorithm, Zwick's algorithm~\cite{zwick02} and the algorithm of Alon et al.~\cite{AGM97} all become $\tilde O(n^{2.5})$.

%

\section{Preliminaries}

\paragraph{Hops and $\delta$-regularity.}
If a path $P$ in $G$ has $\ell$ edges then $P$ is said to have $\ell$ \emph{hops}. 
Let $A^{\leq \ell}$ be a matrix where $a^{\leq \ell}_{ij}$ is the weight of the shortest path from $v_i$ to $v_j$ that has at most $\ell$ hops.

\begin{definition}[$\delta$-regularity~\cite{AGM97}]
For any positive integer $\delta$, a weighted adjacency matrix $A$ of a graph $G=(V,E,w)$ with $V=\{v_1,v_2,\ldots,v_n\}$ is \emph{$\delta$-regular} if for every pair of vertices $v_i,v_j\in V$:
(i) if $a^*_{ij} \neq -\infty$ then $a^*_{ij} = a^{\le \delta}_{ij}$,
 and
(ii) if $a^*_{ij} = -\infty$ then $a^{\le \delta}_{ij}<0$.


\end{definition}

The following Lemma was proven in~\cite{AGM97}.
\begin{lemma}[Lemma 2 in \cite{AGM97}]\label{lem:n-squared-regular}
Any weighted adjacency matrix $A \in {\{-1,0,1,\infty\}}^{n\times n}$ is $n^2$-regular.
\end{lemma}

We emphasize that for a given $\delta$-regular matrix $A$, 
the entries of $A^{\le \ell}$ are not necessarily the same as $A^*$.
Specifically, since the shortest path from $v_i$ to $v_j$ in $G$ could have more than $\ell$ hops, it is possible that $a^{\le \ell}_{ij} \gg a^*_{ij}$.
However, if $A$ is a $\delta$-regular matrix, then  if there exists a shortest path in $G$ between a pair of vertices $u$ and $v$ that does not contain a negative cycle, then there exists a shortest path in $G$ between $u$ and $v$ with at most $\delta$ hops.


\section{Reducing $\{-1,0,1\}$-APSP to \tminmaxShortcut{}}\label{sec:reduction}

\subsection{Canonical Graphs.}
We start with the following definition of a canonical graph which is implicit in~\cite{AGM97}:

\begin{definition} [Canonical graph]
 Let  $G=(V,E,w)$ be a \weightdirect{}  graph with $w:E\rightarrow \{-1,0,1\}$. A  \weightdirect{} graph $G'=(V,E', w')$ is a \canonical{} graph of $G$ if $w':E'\rightarrow \{-1,0,1\}$
 and for any shortest path $P$ from $v_i$ to $v_j$ in $G$  there exists a shortest path $P'$ from $v_i$ to $v_j$ in $G'$ that satisfies the following conditions:
 \begin{enumerate}
     \item $w(P) = w'(P')$.
    \item If $P$ has exactly $\ell$ hops then $P'$ has at most $\ell$ hops.
    \item If $P'$ is not a single edge then $P'$ does
    not contain zero weight edges.
 \end{enumerate}
\end{definition}
Notice that due to the first condition in the definition of canonical graphs,  the distance matrices of a graph $G$ and its canonical graph $G_c$ are the same.
The following lemma, which is proven in~\cite{AGM97}, states that canonical graphs can be efficiently constructed.

\begin{lemma}[Lemma 5 in \cite{AGM97}]\label{lem:\canonical-costruction}
There exists an algorithm which constructs a \canonical{} graph of a \weightdirect{} graph $G=(V,E,w)$ where $w:E\rightarrow \{-1,0,1\}$ in $O(n^\omega)$.
\end{lemma}


\subsection{The Reduction Algorithm} \label{alg:APSP}
We prove Theorem~\ref{thm:main} by presenting
algorithm  \AlgoAPSPtoMINMAX{A}{\delta} which solves the $\{-1,0,1\}$-APSP problem on a graph whose weighted adjacency matrix $A$ is guaranteed to be $\delta$-regular, for an integer $\delta>0$.
Notice that, by Lemma~\ref{lem:n-squared-regular}, if $A$ is the weighted adjacency matrix of $G$ then   $A$ is $n^2$-regular, and therefore \AlgoAPSPtoMINMAX{A}{n^2} solves $\{-1,0,1\}$-APSP on $G$.

We now turn to  describe \AlgoAPSPtoMINMAX{A}{\delta}. (See Algorithm~\ref{palg:main} for a pseudocode).
The input is an integer value $\delta>0$ and a $\delta$-regular weighted adjacency matrix $A$.
Let  $G=(V,E,w)$ be the graph represented by weighted adjacency matrix $A$.
If $\delta=1$, and so $A$ is $1$-regular, the following lemma allows to compute $A^*$ in $O(n^\omega)$ time.

\begin{lemma}[Corollary of Lemmas 3 and 4 from~\cite{AGM97}]\label{lem:one-regular}
There exists an algorithm which solves the $\{-1,0,1\}$-APSP problem on graphs whose weighted adjacency matrix is 1-regular in $O(n^\omega)$ time.
\end{lemma}

Thus, suppose $\delta >1$.
The algorithm begins by constructing a \canonical{}
graph $G_c=(V,E_c,w_c)$ for $G$. Let $C$ be the weighted adjacency matrix of $G_c$ and let $T=\lceil C^{\leq2}/2 \rceil$. Notice that $T\in {\{-1,0,1,\infty\}}^{n\times n}$ and computing $T$ costs $O(n^\omega)$ time.

It was shown in~\cite{AGM97} that  $T$ is a $\lceil\delta/2\rceil$-regular weighted adjacency matrix and that the distance matrix $T^*$ of the graph represented by $T$ is $\lceil C^*/2 \rceil$.
Next, the algorithm performs a recursive call
\AlgoAPSPtoMINMAX{T}{\lceil\delta/2\rceil}, to compute  $T^*$.
From the definition of \canonical{} graphs it follows that $C^*=A^*$. Hence,  $T^*=\lceil C^*/2 \rceil=\lceil A^*/2 \rceil$.

\begin{algorithm}[t]
\caption{\AlgoAPSPtoMINMAX{A}{\delta}}
\begin{algorithmic}[1]

\IF{$\delta = 1$}\label{line:one-regular-begin}
\STATE compute the distance matrix $A^*$ using Lemma \ref{lem:one-regular}
\RETURN $A^*$
\ENDIF \label{line:one-regular-end}

\STATE Let $C$ be the  weighted adjacency matrix of the \canonical{} graph for the graph that $A$ represents.\label{line:5}
\STATE $T^* \gets$ \AlgoAPSPtoMINMAX{\lceil C^{\leq2}/2 \rceil}{\lceil \delta/2\rceil}.\label{line:6}

\STATE $x^{-}_{ij} \gets
  \begin{cases}
        -\infty, & \text{if}\ c_{ij} =-1\\
        \infty , &  \text{otherwise}
    \end{cases}$
\newline $x^{+}_{ij} \gets
  \begin{cases}
        -\infty, & \text{if}\ c_{ij} =+1\\
        \infty , &  \text{otherwise}
    \end{cases}$ \label{line:7}

\STATE Let $M$ be a matrix where $m_{ij}=t^*_{ij}-1$
\newline $Z^{-}\gets T^*\ovee_{T^*} X^{-}$
\newline $Z^{+}\gets T^*\ovee_{M} X^{+}$ \label{line:8}

\STATE $ a^*_{ij}\gets
    \begin{cases}
      t^*_{ij}, & \text{if}\ t^{*}_{ij} = \pm \infty \\
      2t^*_{ij}-1, & \text{if}\ z^{+}_{ij} = 1 \text{ or } z^{-}_{ij} = 1\\
      2t^*_{ij}, & \text{otherwise}
    \end{cases}$
    \label{line:9}

\RETURN $A^*$

\end{algorithmic}
\label{palg:main}

\end{algorithm}

Consider a pair of vertices $v_i$ and $v_j$.
Since $t^*_{ij} = \lceil a^*_{ij}/2 \rceil$ it follows that if  $a^*_{ij}$ is even then $2 t^*_{ij}=a^*_{ij}$ and if $a^*_{ij}$ is odd then $2 t^*_{ij} = a^*_{ij} + 1$.
Therefore, given the value of $t^*_{ij}$ and the parity of $a^*_{ij}$, the algorithm is able to compute the value of $a^*_{ij}$.
Since $T^*$ is computed by the recursive call, the only remaining task in order to compute $A^*$ is to compute the parity of $a^*_{ij}$ for every $i,j\in [n]$.






\paragraph{Establishing parity.}
We now describe how to determine the parity of the entries of $A^*$ without direct access to $A^*$, but with access to $T^*$ and $A$.  Notice that this is where our algorithm differs from  the algorithm of~\cite{AGM97}.

Let
\begin{equation*}
x^{+}_{ij} =
  \begin{cases}
        -\infty, & \text{if}\ c_{ij} =+1\\
        \infty , &  \text{otherwise}
        \end{cases}
\end{equation*}
\begin{equation*}
x^{-}_{ij} =
  \begin{cases}
        -\infty, & \text{if}\ c_{ij} =-1\\
        \infty , &  \text{otherwise}
        \end{cases}
\end{equation*}
and define a matrix $M$ where $m_{ij} = t^*_{ij}-1$.

The algorithm computes
\begin{align*}
Z^{+}& =T^*\ovee_{M} X^{+} \\
Z^{-}& =T^*\ovee_{T^*} X^{-}.
\end{align*}

The following three lemmas demonstrate the connection between both $Z^+$ and $Z^-$, and the parity of entries in $A^*$.
For the proofs of these lemmas, notice that $t^*_{ij} \ne \infty$ if and only if there exists some path from $v_i$ to $v_j$ in $G_c$ (and so, by the definition of a canonical graph, there must exists a path from $v_i$ to $v_j$ in $G$). The following observation is due to $m_{ij} = t^*_{ij}-1$.

\begin{observation}
$t^*_{ij} = \pm\infty$ if and only if $c^*_{ij}=a^*_{ij}=m_{ij}=t^*_{ij}$.
\end{observation}


\begin{lemma}\label{lem:plus-odd}

Let $v_i,v_j\in V$ and assume that $m_{ij}\ne \infty$, then $z^{+}_{ij}=1$ if and only if there exists a shortest path $P$ in $G_c$ from $v_i$ to $v_j$ such that the weight of the last edge of $P$ is $1$ and $a^*_{ij}$ is odd.
\end{lemma}

\begin{proof}

The proof is a case analysis showing that the only case in which $z^{+}_{ij}= 1$ is when there exists a shortest path from $v_i$ to $v_j$ whose last edge has weight 1 and $a^*_{ij}$ is odd, and that in such a case, it must be that $z^{+}_{ij}= 1$.

Notice that, since $m_{ij} \ne \infty$ and by the definition of $X^+$, if there exists a path from $v_i$ to $v_j$ in $G_c$ whose last edge has weight $1$, then there exists a value $\hat k\in [n]$ such that $t^*_{i\hat k} = \min_{k=1}^{n} \{\max(t^*_{ik}, x^+_{kj})\}\ne \infty$.
Moreover, since $t^*_{i\hat k} \ne \infty$, there exists a path from $v_i$ to $v_j$ in $G_c$ whose last edge has weight $1$ and the vertex preceding $v_j$ on this path is $v_{\hat k}$.
Therefore, since $C^*=A^*$ and by the triangle inequality, $2t^*_{ij}-1 \le c^*_{ij} \le c^*_{i\hat k}+1 \le 2t^*_{i\hat k} +1$, and so $m_{ij}= t^*_{ij} -1 \le t_{i\hat k} = \min_{k=1}^{n} \{\max(t^*_{ik}, x^+_{kj})\} $.
Moreover, notice that, by  Definition~\ref{def:t-minmax}, $z^{+}_{ij}= 1$ if and only if  $\min_{k=1}^{n} \{\max(t^*_{ik}, x^+_{k j})\} = m_{ij}$.


Suppose that there is no path from $v_i$ to $v_j$ whose last edge has weight 1. Then for every $1\le k \le n$ we have that either $t^*_{ik} =\infty$ or $x^+_{k j} =\infty$ (since otherwise there is a path from $v_i$ to $v_k$ in $G_c$ and $c_{kj}=1$).
Thus, $\min_{k=1}^{n} \{\max(t^*_{ik}, x^+_{k j})\} =\infty$.
Since we assumed that $m_{ij}\ne \infty$ then  it is guaranteed that in this case $z^{+}_{ij}= 0$.

\begin{figure}
\centering
\includegraphics[width=0.65\textwidth]{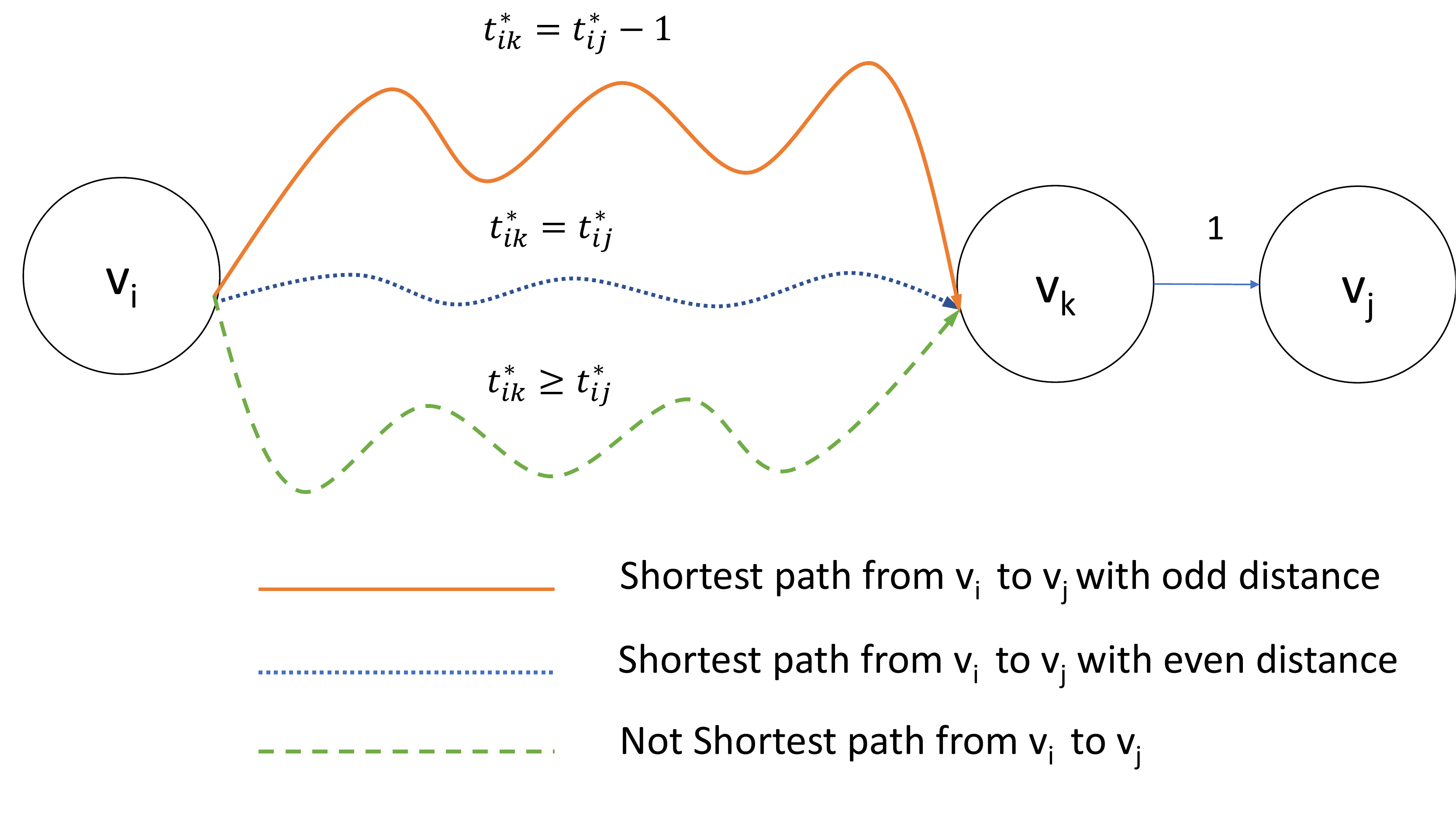}
\caption{Caption goes here.}\label{fig:example-weight-1}
\end{figure}

Thus, for the rest of the proof we assume that there exists a path $P$ from $v_i$ to $v_j$ in $G_c$ whose last edge is $(v_k,v_j)$, $c_{kj}= 1$, and without loss of generality the prefix of $P$ from $v_i$ to $v_k$ is a shortest path.
The following three categories cover all the possibilities for  $P$ (See Figure~\ref{fig:example-weight-1}):
\begin{enumerate*}[label=\roman*,align=CenterWithParen]
    \item $P$ is a shortest path and $a^*_{ij}$ is odd.
    \item $P$ is a shortest path and $a^*_{ij}$ is even.
    \item $P$ is not a shortest path.
\end{enumerate*}

\paragraph{First category.}
Consider  the case where the  weight of the last edge $(v_k,v_j)$ of $P$ is $1$ and $a^*_{ij}$ is odd.
Since $a^*_{ij}$ is odd we have $a^*_{ij}=2t^*_{ij}-1$. Since $a^*_{ik}=a^*_{ij}-1$ we have that $a^*_{ik}$ is even, implying  that $a^*_{ik} = 2t^*_{ik}$. Therefore, $2t^*_{ik}= a^*_{ik}=a^*_{ij}-1=2t^*_{ij} - 2$, and
$t^*_{ik} = t^*_{ij}-1 = m_{ij}$.

\paragraph{Second category.}
We now turn to the case where the  weight of the last edge $(v_k,v_j)$ of $P$ is $1$ and $a^*_{ij}$ is even.
Since $a^*_{ij}$ is even, we have $a^*_{ij}=2t^*_{ij}$.
Since $a^*_{ik}=a^*_{ij}-1$ then $a^*_{ik}$ is odd, implying that
$a^*_{ik}=2t^*_{ik}-1$. Therefore,  $2t^*_{ik}-1 = a^*_{ik}=a^*_{ij}-1=2t^*_{ij}-1$, and $t^*_{ik} = t^*_{ij} > m_{ij}$.

\paragraph{Third category.}
Next, consider the case where $P$ is not a shortest path. Thus, $a^*_{ik} +1 > a^*_{ij}$.
Recall that  $2t^*_{ij}-1 \leq a^*_{ij}$
and $a^*_{ik}\leq 2t^*_{ik}$.
Thus, $2t^*_{ij}-1 \leq a^*_{ij} <a^*_{ik}+1 \leq 2t^*_{ik}+1$
and so $t^*_{ik} > t^*_{ij}-1 = m_{ij}$.

\paragraph{Conclusion.}
Notice that when evaluating $\min_{k=1}^{n} \{\max(t^*_{ik}, x^+_{k j})\}$, the term $\max(t^*_{ik}, x^+_{k j}) \ne \infty$ only for values of $k$ for which there exists a path from $v_i$ to $v_k$ in $G_c$ and $c_{kj}=1$.
Thus, since we assume that $m_{ij}\ne \infty$, the only values of $k$ which need to be considered are values of $k$ such that there exists a path from $v_i$ to $v_j$ whose last edge is $(v_k,v_j)$ and $c_{kj}=1$; all other values of $k$ have $\max(t^*_{ik}, x^+_{k j}) = \infty$ and are therefore not relevant.

If all such paths fall into either the second or third category, then for every such path $P$ whose last edge is $(v_k,v_j)$ we have $\max(t^*_{ik},x^+_{k j}) = t^*_{ik} > m_{ij}$, and so $\min_{k=1}^{n} \{\max(t^*_{ik}, x^+_{k j})\} > m_{ij}$ implying $z^+_{ij} = 0$.
On the other hand, if there exists at least one such path that falls into the first category, then for any such path $P$ whose last edge is $(v_k, v_j)$, we have $\max(t^*_{ik},x^+_{k j}) = t^*_{ik} = t^*_{ij} -1 = m_{ij}$.
Since we always have $m_{ij} \leq  \min_{k=1}^{n} \{\max(t^*_{ik}, x^+_{kj})\}$, it follows that
$\min_{k=1}^{n} \{\max(t^*_{ik},x^+_{k j})\} = m_{ij}$, and so $z^+_{k j} =1$.
\end{proof}

\begin{lemma}\label{lem:minus-odd}
Let $v_i,v_j\in V$ and assume that $m_{ij}\ne \infty$, then $z^{-}_{ij}=1$ if and only if there exists a shortest path $P$ from $v_i$ to $v_j$ in $G_c$ such that the weight of the last edge of $P$ is $-1$ and $a^*_{ij}$ is odd.
\end{lemma}

The proof of Lemma~\ref{lem:minus-odd} is similar to the proof of Lemma~\ref{lem:plus-odd}, and so the proof is deferred to Appendix~\ref{app:minus-odd}.

\begin{lemma}\label{lem:0-edge}
Let $v_i,v_j\in V$. If $a^*_{ij}$ is odd then there exists a shortest path from $v_i$ to $v_j$ in $G_c$ that has a last edge with weight either $1$ or $-1$.
\end{lemma}
\begin{proof}
Assume by contradiction that $a^*_{ij}$ is odd, but the last edge in every shortest path from $v_i$ to $v_j$  has weight 0.
Since $G_c$ is a \canonical{} graph, for every pair of vertices $v_i$ and $v_j$,
there must exist a shortest path $P\in G_c$ from $v_i$ to $v_j$ that either has only one edge or does not contain any edges with weight $0$.
Thus, if all of the shortest paths from $v_i$ to $v_j$ have a last edge of weight $0$, then $a^*_{ij}=0$, and so $a^*_{ij}$ is even, which is a contradiction.
\end{proof}
By Lemma \ref{lem:0-edge}, for every pair of vertices $v_i$ and $v_j$, $a^*_{ij}$ is odd if and only if there exists a shortest path from $v_i$ to $v_j$ in $G_c$ such that the weight of the last edge of the path is $1$ or $-1$. Moreover, by Lemmas~\ref{lem:plus-odd} and~\ref{lem:minus-odd}, if the weight of the last edge of a shortest path from $v_i$ to $v_j$ in $G_c$ is $1$ then $z^+_{ij} = 1$, and if the weight of the last edge of a shortest path from $v_i$ to $v_j$ in $G_c$ is $-1$ then $z^-_{ij} = 1$.

\begin{corollary}\label{cor:a-is-even}
Let $v_i, v_j\in V$. Then $a^*_{ij}$ is odd if and only if either $z^{-}_{ij}=1$ or $z^{+}_{ij}=1$.
\end{corollary}

By Lemmas~\ref{lem:plus-odd} and~\ref{lem:minus-odd}, and by Corollary~\ref{cor:a-is-even},  for each pair of vertices $v_i$ and $v_j$ we have that
 \begin{equation*}
    a^*_{ij}=
    \begin{cases}
    t^*_{ij}, & \text{if}\ t^{*}_{ij} = \pm \infty \\
        2t^*_{ij}-1, & \text{if}\
       z^{+}_{ij} = 1 \vee z^{-}_{ij} = 1\\
      2t^*_{ij}, & \text{otherwise}
    \end{cases}
  \end{equation*}
Thus, constructing $A^*$ from $T^*$, $Z^+$, and $Z^-$ costs $O(n^2)$ time.

\paragraph{Time cost.}
We now analyze the run time of \AlgoAPSPtoMINMAX{A}{n^2} where $A$ is the weighted  adjacency matrix of $G$.
In the following, the line numbers refer to the lines in Algorithm~\ref{palg:main}.
Let $f(n, \delta)$ denote the run time of the algorithm on the weighted adjacency matrix of a $\delta$-regular graph with $n$ vertices.
By Lemma \ref{lem:one-regular}, the computation in Lines~\ref{line:one-regular-begin}--\ref{line:one-regular-end} costs $O(n^{\omega})$ time.
 By Lemma \ref{lem:\canonical-costruction}, the construction of the \canonical{} graph in Line \ref{line:5}, costs $O(n^{\omega})$ time.
In Line \ref{line:6} the algorithm first computes the matrix $\lceil C^{\leq 2}/2\rceil$ in $O(n^{\omega})$ time (using FMM), and then the algorithm makes a recursive call that costs $f(n,\delta/2)$ time.
Lines \ref{line:7} and \ref{line:9} cost $O(n^2)$ time.
Line \ref{line:8} has two \tminmaxShortcut{} computations which cost $O(T_{TMM}(n))$ time.
Thus, $$f(n,\delta) = f(n,\lceil\delta/2\rceil) + O(n^{\omega}+T_{TMM}(n)).$$
Finally, since in each recursive $\delta$ is halved, the number of recursive calls is $\log \delta$.
Thus, $f(n,n^2) = \tilde{O}(n^{\omega}+T_{TMM}(n))$.

\section{A Simple Algorithm for \Name}\label{sec:restricted-algo}
Recall that Algorithm~\ref{palg:main} (of Theorem~\ref{thm:main}) makes use of an algorithm for \tminmaxShortcut{}.
However, these calls are applied to a \emph{restricted} family of inputs: the entries in the second matrix $B$ are $\pm \infty$, and the target matrix $T$ has the property that for any $i,j\in [n]$, $t_{ij}\le \min_{k=1}^{n} \{\max(a_{ik}, b_{kj})\}$.
In this section we prove Theorem~\ref{thm:restricted} by describing a simple algorithm for \tminmaxShortcut{} when the inputs are \emph{restricted} to the family of inputs that are seen when calling \tminmaxShortcut{} in Algorithm~\ref{palg:main}; see Definition~\ref{def:restricted-T-minmax}.
Our algorithm combines \emph{fast rectangular Boolean matrix multiplication} (BMM) with a \emph{heavy-light} type of decomposition
\footnote{According to Definition \ref{def:restricted-T-minmax} $A$ and $T$, are both from $\mathbb{R}^{n\times n}$, while in \AlgoAPSPtoMINMAX{}{} the matrices can have entries from $\{-\infty,\infty\}$.
However, if $t_{ij}\in\{-\infty,\infty\}$ then $a^*_{ij}$ is set to be $t_{ij}$ regardless the product.}.

For every $i \in [n]$, let $A_i$ be the $i^{th}$ row of $A$.
The algorithm sorts the pairs $(a_{ij},j)$, where $ j\in [n]$, in a lexicographically increasing order and stores the result in array $L_i$.
For $A_i$, we say that a value that occurs more than $n^t$ times in $A_i$ is \emph{heavy} for $A_i$, and
a value in $A_i$ that occurs at most $n^t$ times in $A_i$ is \emph{light} for $A_i$.
It is straightforward to partition the values in $A_i$ to heavy and light values in $O(n)$ time.
Notice that there are at most $n/n^t = n^{1-t}$ heavy elements for $A_i$.

The algorithm constructs a rectangular  matrix $H$ as follows.
For every $i\in [n]$ and every heavy value $x$ in $A_i$
the algorithm adds a row to $H$ that corresponds to $x$.
Let $\rho(i,x)$ be the index of the row in $H$ that is added for heavy value $x$.
Since there are at most $n^{1-t}$ heavy values for each $A_i$, the number of rows in $H$ is at most $n^{2-t}$.
For $j\in [n]$, the $j^{th}$ entry in the $\rho(i,x)$'th row is set to $1$ if and only if the $j^{th}$ entry of $A_i$ contains the value $x$; otherwise  the $j^{th}$ entry is set to $0$.
Thus, $H$ is a Boolean matrix of size $O(n^{2-t})\times n$.

The algorithm converts the matrix $B$ into a Boolean matrix $B'$ as follows. For every $ i,j \in [n]$ let:
\begin{equation*}
  b'_{ij} =
  \begin{cases}
        1 & \text{if}\ b_{ij} = -\infty\\
        0 &  \text{otherwise}
        \end{cases}
\end{equation*}

Next, the algorithm computes the rectangular BMM $F=H B'$.
Finally, the algorithm constructs matrix $C$ as follows. Consider the value $t_{ij}$ from the target matrix $T$, where  $i,j \in[n]$. If $t_{ij}$ is heavy in row $A_i$ then the algorithm sets $c_{ij}$ to $f_{qj}$ for $q=\rho(i,t_{ij})$.
If $t_{ij}$ is light, then the algorithm uses $L_i$ to access all the occurrences of $t_{ij}$ in $A_i$.
Let $I=\{ k_1,\ldots,k_\ell \}$ be the set of  all  indices of columns in row $i$ that contain the value $t_{ij}$. Thus, for every $k\in I$ we have $a_{ik}=t_{ij}$.
If there is $k\in I$ such that $b_{kj}=-\infty$ the algorithm sets $c_{ij}$ to $1$, otherwise the algorithm sets $c_{ij}$ to $0$.

\paragraph{Correctness.}
We now show that the matrix $C$ computed above equals $A \ovee_T B$.
By Definition~\ref{def:restricted-T-minmax}  $c_{ij}$ should be $1$ if  $t_{ij} = \min_{k=1}^{n} \{\max(a_{ik}, b_{kj})\}$, and $0$ otherwise. By Definition~\ref{def:restricted-T-minmax}   we also have  $t_{ij} \leq  \min_{k=1}^{n} \{\max(a_{ik}, b_{kj})\}$.
Thus, if there exists a $k$ such that $t_{ij} = \max(a_{ik}, b_{kj}) $ then $t_{ij} =  \min_{k=1}^{n} \{\max(a_{ik}, b_{kj})\}$.
Notice that if there exists $k'$ such that $t_{ij} > \max(a_{ik'}, b_{k'j})$, then $t_{ij} >  \min_{k=1}^{n} \{\max(a_{ik}, b_{kj})\}$, which is a contradiction.
Moreover, if for every $k\in[n]$,  $t_{ij} \neq \max(a_{ik}, b_{kj})$, then $t_{ij} \neq \min_{k=1}^{n} \{\max(a_{ik}, b_{kj})\}$.
Since $B\in \{-\infty,\infty\}^{n\times n}$, then we conclude that there exits $k\in [n]$ such that $t_{ij} = \max(a_{ik}, b_{kj})$
if and only if $a_{ik} = t_{ij}$ and $b_{kj} = -\infty$.

If $t_{ij}$ is light in row $A_i$ then for every $k\in [n]$ where $a_{ik}=t_{ij}$ the algorithm checks whether $b_{kj}=-\infty$ or not. Thus if  $t_{ij} = \max(a_{ik}, b_{kj})$ and $t_{ij}$ is light then the algorithm will detect this case.
If $t_{ij}$ is heavy in row $A_i$ then row $\rho(i,x)$ in $H$ has $1$ in the same columns that contain $t_{ij}$ in row $A_i$ and $0$ in  all other columns. In matrix $B'$ there are $1$ values in all entries that correspond to entries in $B$ with value $-\infty$ and $0$ in  all other locations.
Hence, in $F$, $f_{qj}=1$ for $q=\rho(i,t_{ij})$ if and only if there exists a $k\in [n]$ such that $a_{ik} = t_{ij}$ and $b_{kj} = -\infty$ which in turn implies that
$t_{ij} = \max(a_{ik}, b_{kj})$.

\paragraph{Time cost.}
The cost of handling a light value is $O(n^t + \log n)$.
Since $T$ contains at most $O(n^2)$ light values, the total cost for handling all light values is at most $O(n^{2+t})$.
Computing $F$ using fast rectangular matrix multiplication takes $O(n^{\omega(2-t,1,1)})$ time.
Thus, the total time cost of the algorithm is $O(n^{2+t} +n^{\omega(2-t,1,1)})$.

\bibliographystyle{plain}
\bibliography{bib}

\newpage
\appendix

\section{Proof of Lemma~\ref{lem:minus-odd}}\label{app:minus-odd}

\begin{proof}[Proof of Lemma~\ref{lem:minus-odd}]
The proof is a case analysis showing that the only case in which $z^{-}_{ij}= 1$ is when there exists a shortest path from $v_i$ to $v_j$ whose last edge has weight $-1$ and $a^*_{ij}$ is odd, and that in such a case, it must be that $z^{-}_{ij}= 1$.

Notice that, since $t^*_{ij} \ne \infty$ and by the definition of $X^-$, if there exists a path from $v_i$ to $v_j$ in $G_c$ whose last edge has weight $-1$, then there exists a value $\hat k\in [n]$ such that $t^*_{i\hat k} = \min_{k=1}^{n} \{\max(t^*_{ik}, x^-_{kj})\}\ne \infty$.
Moreover, since $t^*_{i\hat k} \ne \infty$, there exists a path from $v_i$ to $v_j$ in $G_c$ whose last edge has weight $-1$ and the vertex preceding $v_j$ on this path is $v_{\hat k}$.
Therefore, since $C^*=A^*$ and by the triangle inequality, $2t^*_{ij}-1 \le c^*_{ij} \le c^*_{i\hat k}-1 \le 2t^*_{i\hat k} -1$, and so $t^*_{ij}  \le t_{i\hat k} = \min_{k=1}^{n} \{\max(t^*_{ik}, x^-_{kj})\} $.
Moreover, notice that, by  Definition~\ref{def:t-minmax}, $z^{-}_{ij}= 1$ if and only if  $\min_{k=1}^{n} \{\max(t^*_{ik}, x^-_{k j})\} = t^*_{ij}$.


Suppose that there is no path from $v_i$ to $v_j$ whose last edge has weight $-1$. Then for every $1\le k \le n$ we have that either $t^*_{ik} =\infty$ or $x^-_{k j} =\infty$ (since otherwise there is a path from $v_i$ to $v_k$ and $c_{kj}=-1$).
Thus, $\min_{k=1}^{n} \{\max(t^*_{ik}, x^-_{k j})\} =\infty$.
Since we assumed that $t^*_{ij}\ne \infty$ then  it is guaranteed that in this case $z^{-}_{ij}= 0$.

Thus, for the rest of the proof we assume that there exists a path $P$ from $v_i$ to $v_j$ in $G_c$ whose last edge is $(v_k,v_j)$, $c_{kj}=-1$, and without loss of generality the prefix of $P$ from $v_i$ to $v_k$ is a shortest path.
The following three categories cover all the possibilities for  $P$:
 (i) $P$ is a shortest path and $a^*_{ij}$ is odd.
(ii)  $P$ is a shortest path and $a^*_{ij}$ is even.
(iii) $P$ is not a shortest path.

\paragraph{First category.}
Consider  the case where the  weight of the last edge $(v_k,v_j)$ of $P$ is $-1$ and $a^*_{ij}$ is odd.
Since $a^*_{ij}$ is odd we have $a^*_{ij}=2t^*_{ij}-1$. Since $a^*_{ik} = a^*_{ij} +1$ we have that $a^*_{ik}$ is even, implying  that $a^*_{ik} = 2t^*_{ik}$. Therefore, $2t^*_{ik}= a^*_{ik}=a^*_{ij}+1=2t^*_{ij} $, and
$t^*_{ik} = t^*_{ij}$.

\paragraph{Second category.}
We now turn to the case where the  weight of the last edge $(v_k,v_j)$ of $P$ is $-1$ and $a^*_{ij}$ is even.
Since $a^*_{ij}$ is even, we have $a^*_{ij}=2t^*_{ij}$.
Since $a^*_{ik}=a^*_{ij}+1$ then $a^*_{ik}$ is odd, implying that
$a^*_{ik}=2t^*_{ik}-1$. Therefore,  $2t^*_{ik}-1 = a^*_{ik}=a^*_{ij}+1=2t^*_{ij}+1$, and $t^*_{ik} = t^*_{ij}+1 > t^*_{ij}$.

\paragraph{Third category.}
Next, consider the case where $P$ is not a shortest path. Thus, $a^*_{ik} -1 > a^*_{ij}$.
Recall that  $2t^*_{ij}-1 \leq a^*_{ij}$
and $a^*_{ik}\leq 2t^*_{ik}$.
Thus, $2t^*_{ij}-1 \leq a^*_{ij} < a^*_{ik}-1 \leq 2t^*_{ik}-1$
and so $t^*_{ik} > t^*_{ij}$.

\paragraph{Conclusion.}
Notice that when evaluating $\min_{k=1}^{n} \{\max(t^*_{ik}, x^-_{k j})\}$, the term $\max(t^*_{ik}, x^-_{k j}) \ne \infty$ only for values of $k$ for which there exists a path from $v_i$ to $v_k$ in $G_c$ and $c_{kj}=-1$.
Thus, since we assume that $t^*_{ij}\ne \infty$, the only values of $k$ which need to be considered are values of $k$ such that there exists a path from $v_i$ to $v_j$ whose last edge is $(v_k,v_j)$ and $c_{kj}=-1$; all other values of $k$ have $\max(t^*_{ik}, x^-_{kj}) = \infty$ and are therefore not relevant.

If all such paths fall into either the second or third category, then for every such path $P$ whose last edge is $(v_k,v_j)$ we have $\max(t^*_{ik},x^-_{k j}) = t^*_{ik} > t^*_{ij}$, and so $\min_{k=1}^{n} \{\max(t^*_{ik}, x^-_{k j})\} > t^*_{ij}$ implying $z^-_{ij} = 0$.
On the other hand, if there exists at least one such path that falls into the first category, then for any such path $P$ whose last edge is $(v_k, v_j)$, we have $\max(t^*_{ik},x^-_{k j}) = t^*_{ij}$.
Since we always have $t^*_{ij} \leq  \min_{k=1}^{n} \{\max(t^*_{ik}, x^-_{kj})\}$, it follows that
$\min_{k=1}^{n} \{\max(t^*_{ik},x^-_{k j})\} = t^*_{ij}$, and so $z^-_{kj} =1$.
\end{proof}

\section{Time Cost of Theorem~\ref{thm:restricted}}\label{app:run-time}
The goal is to find the value $0\le t \le 1$ which minimizes $O(n^{2+t} +n^{\omega(2-t,1,1)})$. 
This goal is equivalent to finding the value  $1\le k \le 2$ such that $\omega(k,1,1) = 4-k$.
In the following we use the fact that  the function $\omega(k,1,1)$ is convex, so we apply linear interpolation between the values of $\omega(k,1,1)$ for $k=1.3$ and $k=1.4$, which are given in~\cite{GU18}.
Specifically, $\omega(1.3,1,1) = 2.621644 $ and $\omega(1.4,1,1)=2.708400 $.
Therefore, the line connecting $(1.3, 2.621644)$ and $(1.4,2.708400)$ is above the point $(k,\omega(k,1,1))$ that we are searching for. This line is given by the equation $$y=0.86756x+1.493816.$$
Solving $y(\hat k) = 4-\hat k$, we have $\hat k=1.3419156349$. 
Therefore, 
\begin{align*}
  \min_{1\le k\le 2} \{\omega(k,1,1) + 4-k\} \le &  \omega(\hat k,1,1) + 4-\hat k\\
  \le  & y(\hat k) + 4-\hat k\\
  =& y(1.3419156349) + 4- 1.3419156349\\
  = & 2.658043651,
\end{align*}
and so $O(n^{2+t} +n^{\omega(2-t,1,1)}) < O(n^{2.658043651})$.

\end{document}